\documentclass[letterpaper, 10 pt, conference]{ieeeconf}  

\IEEEoverridecommandlockouts                              
\usepackage[utf8]{inputenc}
\usepackage[english]{babel}

\usepackage{cite}
\usepackage{multirow}
\usepackage{graphicx} 
\usepackage[table,xcdraw]{xcolor}
\usepackage{epsfig} 
\usepackage{subfigure}

\usepackage[letterpaper,bindingoffset=0.2in,%
left=0.667in,right=0.667in,top=0.80in,bottom=0.597in,%
footskip=.25in]{geometry}

\usepackage{multicol}

\usepackage{algorithm}
\usepackage[noend]{algpseudocode}

\usepackage{amsthm} 
\usepackage{amsmath} 
\usepackage{amssymb}  

\theoremstyle{plain}

\newtheorem{Lemma}{Lemma}

\newcommand{\cosimo}[1]{\textcolor{black}{#1}}
\newcommand{\fra}[1]{\textcolor{black}{#1}}

\renewcommand{\baselinestretch}{0.9765}

\title{\LARGE \bf \cosimo{Covid-19 and Flattening the Curve: a Feedback Control Perspective}}

\author{Francesco Di Lauro$^{1}$, Istv\'an Zolt\'an Kiss$^{1}$, 
Daniela Rus$^{2}$, Cosimo Della Santina$^{3}$ %
\thanks{$^{1}$ F. Di Lauro and I.Z. Kiss are with Department of Mathematics, University of Sussex, Falmer, Brighton BN1 9QH, UK, and  acknowledge support from the Leverhulme Trust for the Research Project Grant RPG2017-370.
$^{2}$D. Rus is with the MIT Computer Science and Artificial Intelligence Laboratory, Massachusetts Institute of Technology, Cambridge, MA, USA.
$^{3}$C. Della Santina is with the Cognitive Robotics Department, Delft University of Technology, 2628 CD Delft, The Netherlands, and with the Institute of Robotics and Mechatronics, German Aerospace Center (DLR), Oberpfaffenhofen, Germany. Contacts {\tt\footnotesize cosimodellasantina@gmail.com}.}
}

\begin{document}

\maketitle

\begin{abstract}
%
\cosimo{
Many of the control policies that were put into place during  the Covid-19 pandemic had a common goal: to flatten the curve of the number of infected people so that its peak remains under a critical threshold. 
}
\cosimo{This letter considers the challenge of engineering a strategy that enforces such a goal using control theory.}
%
%
\cosimo{
    We introduce a simple formulation of the optimal flattening problem, and provide a closed form solution.
}
%
%
\cosimo{
    This is augmented through nonlinear closed loop tracking of  the nominal solution, with the aim of ensuring close\--to\--optimal performance under uncertain conditions.
}
%
%
\cosimo{A key contribution of this paper is to provide validation of the method with} extensive and realistic simulations in a Covid-19 scenario, with particular focus on the case of Codogno - a small city in Northern Italy that has been among the most harshly hit by the pandemic.
%
%
\end{abstract}

\section{Introduction}\label{sec:introduction}

%
%
\cosimo{Defining  and  implementing  social  distancing  protocols (SD)   is   a significant   challenge   with   economical, political, and scientific considerations. The definition of a clear or optimal goal remains unclear.}
\cosimo{As an example, 
consider the direct reduction of deaths by Covid-19. 
Imposing this goal requires the harshest measures possible, for an indefinite period of time.
According to the available models \cite{Kiss2017} a monotonic relationship exists between this cost function and the SD level.
}
\cosimo{Yet, this strategy has many potential drawbacks.
First,} extreme levels of lockdown are unsustainable in the long run, due to the vast range of pernicious secondary effects (e.g. poverty \cite{goolsbee2020fear}, mental illnesses \cite{bhuiyan2020covid}) 
 \cosimo{which in turn are themselves associated with a rise in mortality. 
}
%
%
%
\cosimo{
Additionally, relaxing or lifting control after a harsh lockdown may lead to a second wave, possibly more critical than the first one \cite{xu2020beware}.
%
}
%
%
%
\cosimo{Another strategy would be to let the epidemic spread freely (red curve in Fig. \ref{fig:actual_cover}) to get herd immunity as fast as possible. 
%
%
This is also hardly acceptable, as it would lead to higher mortality~\cite{armstrong2020outcomes}, and to a prolonged stress of the health care system.
%
}
%
%
\cosimo{
The ``\textit{flattening the curve}'' strategy provides a third option, which promises to combine the benefits of the two extremes \cite{ferguson2020report,thunstrom2020benefits}.
The key idea (of which  Fig. \ref{fig:actual_cover} provides a visual representation) is to allow some level of disease spreading, while ensuring that people seeking medical assistance can access the health care system. 
%
%
%
%
}
%
%

%
%
\begin{figure}
	\centering
	\includegraphics[width = .9\columnwidth]{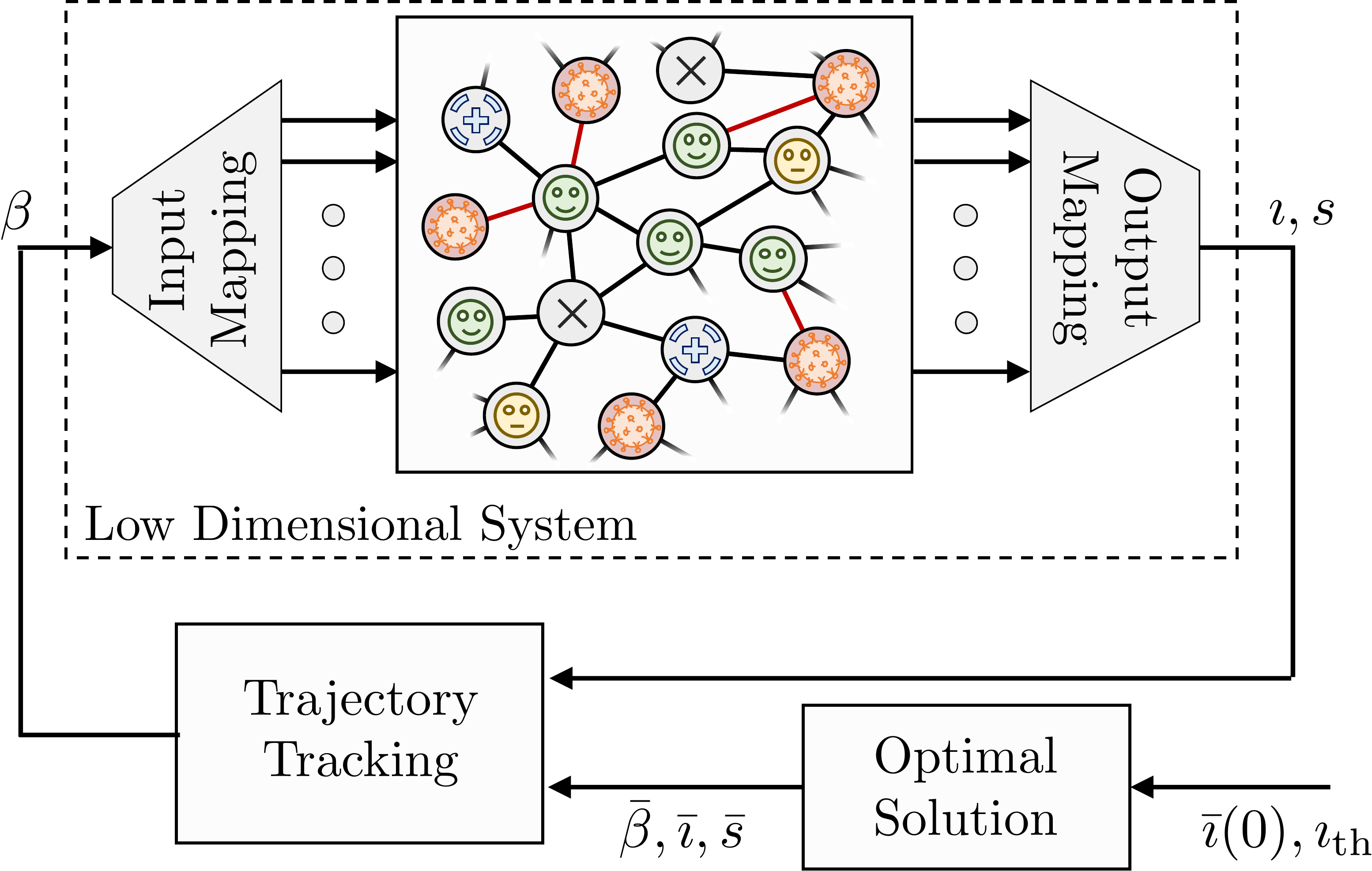}
	\vspace{-0.1cm}
	\caption{Block diagram of the strategy proposed in this paper. 
		The input and output maps reduce the  high\--dimensional dynamics of the outbreak to a simpler evolution of few salient characteristics, namely the prevalence of infected and susceptible $\imath, s$, which are  sensible to  changes in the level of SD, modelled here as different values of the transmission rate of infection $\beta$.
		%
		%
		A nonlinear feedback controller acts within this representation implementing trajectory tracking of an optimal control policy. 
	}
	\label{fig:cover}
\end{figure}
%
%
%
%
%
%
%
%
\cosimo{A vast pre-Covid-19 pandemic literature~\cite{nowzari2016analysis} on designing controllers for dealing with epidemics already exists. However, none of these works tackled the curve flattening goal, since no pandemics before threatened to overburden the healthcare system on such a large scale.
}
%
%
\cosimo{In the context of Covid-19, open loop optimal control is proposed in \cite{di2020timing} for selecting the optimal timing of a time-limited lockdown, and in \cite{djidjou2020optimal} the authors find a trade-off between number of deaths and damage to the economy. 
Yet, feed-forward strategies are quite prone to uncertainties naturally associated with epidemics~\cite{di2020impact}. 
}
\cosimo{More robust strategies have been proposed, relying on feedback control.}
A linear controller is proposed in \cite{giordano2020modelling}. A fast switching strategy with duty cycle selected through a slow feedback is discussed in \cite{bin2020fast}.
In \cite{kohler2020robust}, the loop is closed by periodically re-planning the optimal action, in a model-predictive-control fashion.
%
%
\cosimo{
An explicit formulation of curve flattening is instead provided
in \cite{morris2020optimal}, where an open loop strategy is devised so to optimally reduce the infectious peak. 
An interesting alternative is discussed in \cite{charpentier2020covid}, where a trade-off between the health care and the socio-economic cost of the pandemic is proposed, and the limited capacity level of
intensive care units is imposed as a constraint.
%
%
Both these solutions are open loop.
}
%

\cosimo{
This letter investigates the use of feedback control theory as a tool for engineering an effective curve flattening strategy. We wish to design a simple rule that can be implemented on a local level, without the need of accessing specialized facilities to run complex optimization routines.  
} 
%
%
%
%
%
We perform extensive simulations of epidemics on  networks~\cite{pastor2015epidemic,Kiss2017}, with conditions inspired by real Covid-19 scenarios. 
\cosimo{This is as far as we know the first time that such analysis is carried out for Covid-19 control related research. \fra{We remark that the acceptable level of ``curve flattening'' is to be decided by policy makers, based upon  cost-benefit analysis. However, once an optimal curve has been identified, this letter offers a novel, theoretically-backed strategy that guarantees that the goal of controlling the epidemic curve is achieved.}
}
%

%
%
%
\section{Background: Model of the Epidemics with Dynamic Interventions}\label{Sec:background}
Consider a fixed population of $N$ individuals, and a disease spreading among them, through direct contacts. Each individual can be in either of three states: (i) susceptible, meaning that they can be infected by the pathogen; (ii) infected, meaning that they contracted the pathogen and they can now infect other susceptible  people; (iii) recovered -and therefore immune, or removed. We denote with $S(t), I(t), R(t)$ the number of people at time $t$ who are susceptible, infected or recovered, respectively. We have that  $S(t)+I(t)+R(t) = N$. We can therefore neglect the study of $R$, as its value can always be recovered from $S,I$ and $N$.
If the population is well mixed
, the evolution of the disease can be described by the SIR model 
\begin{equation}\label{eq:full_dynamics}
    \dot{s}(t) = - \beta \imath(t) s(t), \quad
    \dot{\imath}(t) = + \beta \imath(t) s(t) - \gamma \imath(t),
\end{equation}
where $s(t)$ and $\imath(t)$ are the system state, indicating respectively the number of susceptible $S(t)$ and infectious $I(t)$, divided by the total population $N$. 
\cosimo{Note that, despite its simplicity, the SIR model has proven able to match real data when applied to Covid-19~\cite{Saad-Royeabd7343,morris2020optimal,thunstrom2020benefits}, and it is therefore widely used in the literature.}
Without loss of generality, we consider that, at $t = 0$, $s + \imath = 1$.
The constant $\gamma \geq 0$ defines the transition rate from the pool of infected, to the compartment of recovered/removed.
$\beta$ is the rate at which an infected individual makes disease-transmitting contacts with other people. 
When SD policies are imposed, the value of $\beta$ varies, \cosimo{$0<\beta_{\mathrm{min}} \leq \beta \leq \beta_{\mathrm{max}}$}, with $\beta_{\mathrm{min}}$ corresponding to total lockdown.
Therefore $\beta$ is the control input of \eqref{eq:full_dynamics}.

%

\section{Control Strategy}
We propose here a control strategy acting on system \eqref{eq:full_dynamics}. As shown by Fig. \ref{fig:cover}, this architecture is made of two components: (i) an optimal open loop action, and (ii) a feedback controller implementing trajectory tracking. 
%

%
%
%
\subsection{Optimal curve flattening under nominal conditions}\label{sec:optimal}


%
\begin{figure}
	\centering
	\includegraphics[width = .9\columnwidth]{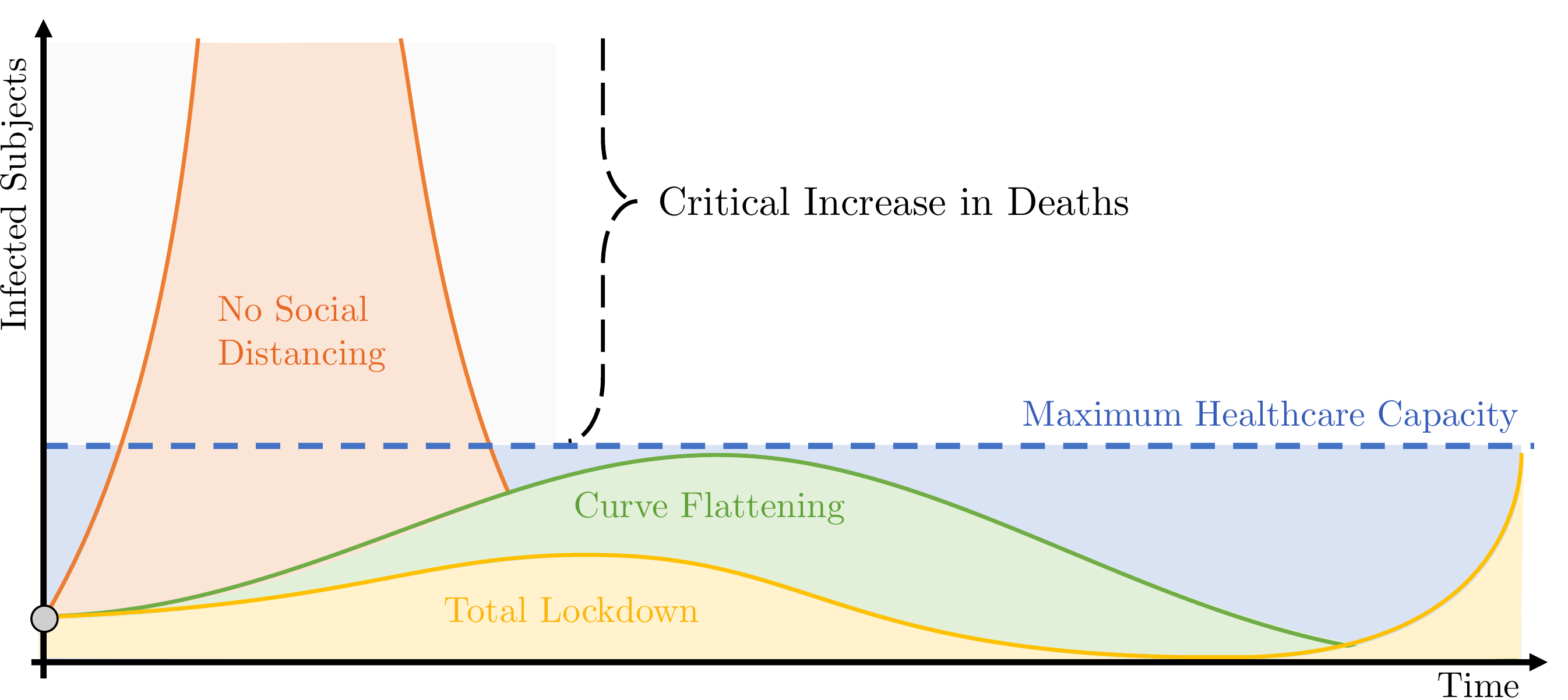}
	\vspace{-0.25cm}
	\caption{
		\cosimo{The aim of this work is to devise a control strategy that achieves the curve flattening goal,  which should result in a curve similar to the green one. 
			The two alternative extreme cases are shown as comparison: the result of no SD is shown in red, and of full lockdown in yellow. }}
	\label{fig:actual_cover}
\end{figure}

%

Our aim here is to introduce a nominal strategy (``Optimal Solution'' in Fig. \ref{fig:cover}) for optimally flattening the epidemic curve $\imath(t)$, so to keep the number of infected people $\imath$ within the maximum capacity of the health care system, $\imath_{th} > 0$.
\cosimo{This can, for instance, be evaluated by considering the percentage of people that will need Intensive Care Units (ICUs), which are probably the most critically limited resources.}
\cosimo{As discussed in the introduction,}
enforcing this constraint is of paramount importance, since exceeding it may provoke a critical failure of the healthcare system, leading to a substantial increase in the number of deaths not only from the disease, but also from uncorrelated health issues.
On the other hand, we want to keep the level of restriction on the population as low as possible, to minimise secondary negative effects. 
%
%
\cosimo{
Note that the curve flatting goal is the result of a careful balance between competing interests, and as such we decide to explicitly impose it as a goal. 
}
We consider the case of a constant $\beta$. 
This simplification is instrumental in making the optimal control problem more manageable. 

We summarize the above considerations through the optimization problem
\begin{equation}\small\label{eq:optimal}
    \max_{\beta \in \mathbb{R}}        \; \beta, \quad
    \text{s.t.} \;\; 0 < \imath(t) \leq \imath_{\mathrm{th}} \; \; \forall t\; \; \text{and} \; \; \eqref{eq:full_dynamics}. 
\end{equation}
%
%
We now propose a Lemma introducing a general solution to this optimal control problem. 

\begin{Lemma}
    The following is the closed form solution of \eqref{eq:optimal}
    \begin{equation}\small\label{eq:closed_form_optimal}
        \beta = -\frac{\gamma}{1 - \imath_{\mathrm{th}} }W_{-1}\left(-\frac{1}{e}\frac{1 - \imath_{\mathrm{th}}}{1 - \imath(0)}\right),
    \end{equation}
    where $W_{-1}$ is the Lambert W function \cite{corless1996lambertw}, branch $-1$.
\end{Lemma}
\begin{proof}
    Since the cost function is linear in the optimization parameter, the optimal value is to be found on the boundary of the feasible set, i.e. $\beta$ has to be such that $\max_t \imath(t) = \imath_{th}$.
    
    The maximum value of $\imath$ is given by the non-trivial solution of $\dot{\imath}(t) = 0$. Combining this condition with the second equation in \eqref{eq:full_dynamics} yields $\displaystyle{{s}^+=\gamma/\beta}$. Further, we can combine the first two lines of~\eqref{eq:full_dynamics} into 
    ${\mathrm{d}\imath}/{\mathrm{d}s} = {\gamma}/{(\beta s)} -1$. 
%
This nonlinear ordinary differential equation can be solved together with the initial condition $s(0) = 1 - \imath(0), \imath(0)$, to get
\begin{equation}\small
\imath(s) = \frac{\gamma}{\beta} \ln\left(\frac{s}{1-\imath(0)}\right) - s + 1.
\label{eq:betanumerical}
\end{equation}
By inverting $\imath(s^+)$ for $\beta$, we get the desired optimal value such that $\max_{t}\imath(t) = \imath_{th}$. The following is a solution for all integer values of $j$,
\begin{equation}\small
        \beta = -\frac{\gamma}{1 - \imath_{\mathrm{th}} }W_{j}\left(-\frac{1}{e}\frac{1 - \imath_{\mathrm{th}}}{1 - \imath(0)}\right),
\end{equation}
\cosimo{where $W_j(\bar{a})$ is the $j\--$th branch of the Lambert W function \cite{corless1996lambertw}. Each of the branches is built as the solution of $\bar{a} = W_j e^{W_j}$.
Among all of them,} only $W_{-1}, W_{0}$ have domain within the real line. Moreover, it is always the case that $W_{0} > W_{-1}$, which in turn assures that the larger value of $\beta$ 
is always reached for $j = -1$, concluding the proof.

\end{proof}

It is worth noting that the argument of  $W_{-1}$ is always between $-1/e$ and $0$ since $0 \ \imath(0) \leq \imath_{\mathrm{th}}$. This is exactly the range of arguments for which the $-1$ branch of the Lambert function is well defined \cite{corless1996lambertw}.
%

\subsection{Trajectory tracking controller}

\begin{figure}
    \centering
    \subfigure[Infected $\imath$]{\includegraphics[height = 0.5\columnwidth,trim={1 0 30 10},clip]{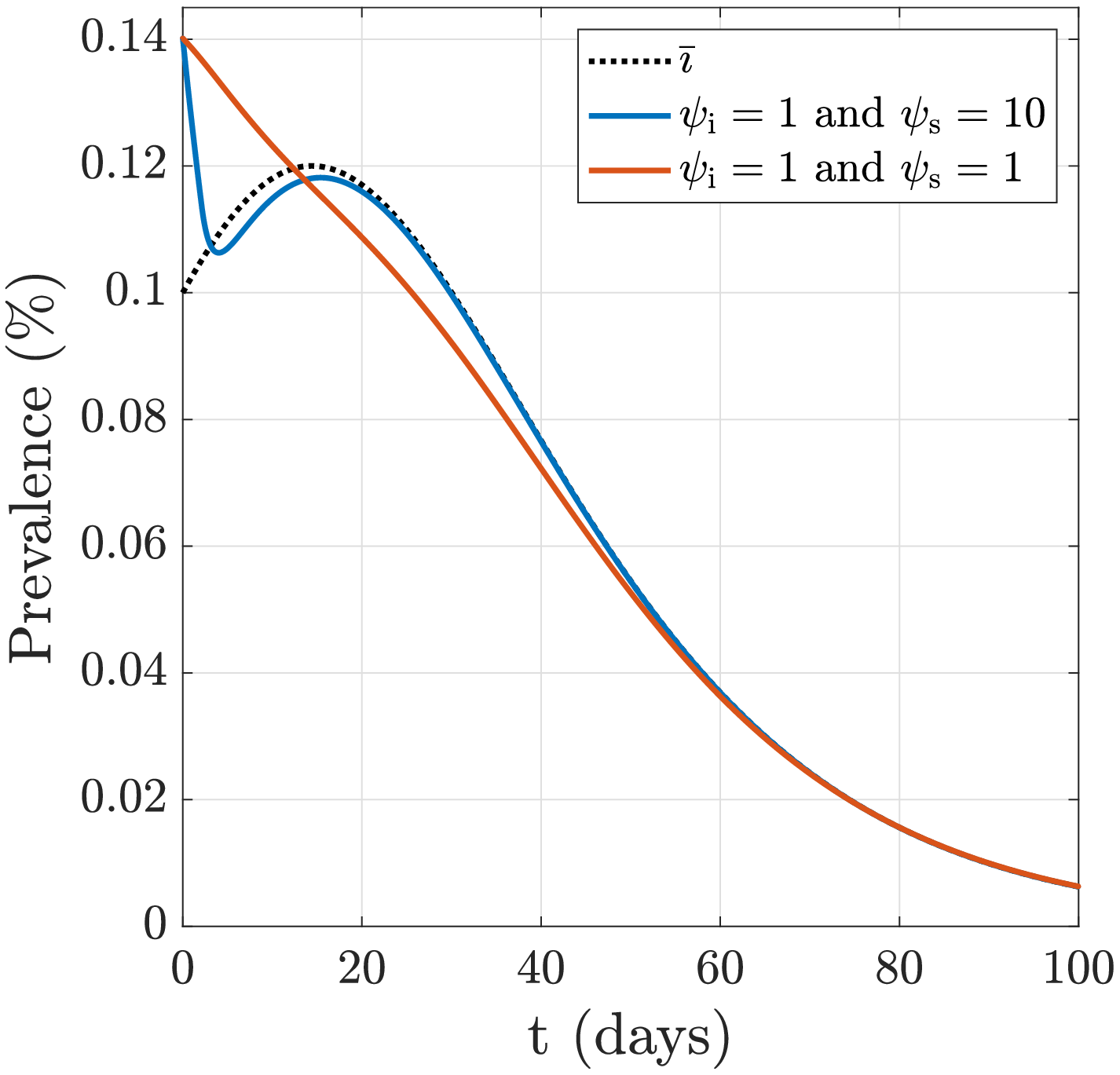}}
    \subfigure[Social Distancing $\beta$]{\includegraphics[height = 0.5\columnwidth,trim={10 0 30 10},clip]{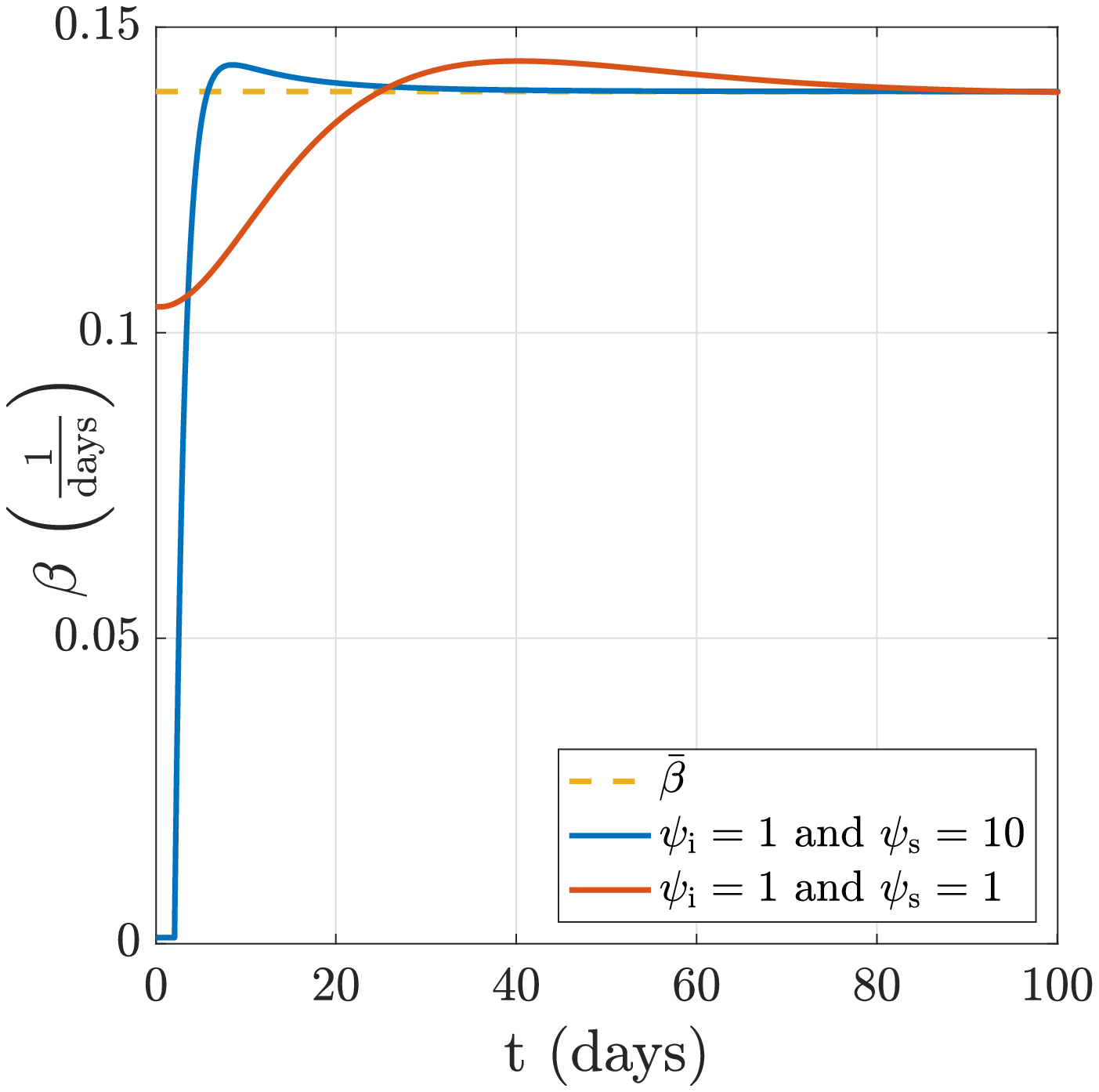}}
    \caption{Two executions of the proposed control architecture when applied to system \eqref{eq:full_dynamics}. Two different choices of control gains $\psi_{\mathrm{i}}$ and $\psi_{\mathrm{s}}$ are considered. The other parameters are $\gamma = 0.1$, $\beta_{\mathrm{max}} = 0.22$, $\bar\imath(0) = 0.1$, $\imath_{\mathrm{th}} = 0.12$, $\imath(0) = 0.14$.
    \cosimo{Susceptibles are not shown for the sake of space.}
    \label{fig:track_traj}
    }
\end{figure}

The following Lemma introduces the tracking controller (``Trajectory Tracking'' in Fig. \ref{fig:cover}) implementing the reactive change of the SD level $\beta$.
Note that 
in principle this controller is agnostic to the choice of the reference to be tracked, and it is introduced as such.

\begin{Lemma}
    The feedback loop composed by the control action
    \begin{equation}\small\label{eq:controller_fb}
    \color{black}
        \beta(s,\imath,t) = 
        + \psi_{\mathrm{i}} (\bar{\imath} - \imath) - \psi_{\mathrm{s}} (\bar{s} - s) +\frac{\bar{s}\,\bar{\imath}}{s \, \imath} \;\bar{\beta}
    \end{equation}
    and the SIR model \eqref{eq:full_dynamics}, is such that \cosimo{$(s,\imath)$ converges exponentially fast to $(\bar{s},\bar{\imath})$,}  
    %
    %
    \cosimo{
    $ \forall  \psi_{\mathrm{i}},\psi_{\mathrm{s}} \in \mathbb{R}, \psi_{\mathrm{s}} > 0, \psi_{\mathrm{i}} \geq 0$,} and if $\bar{s}, \bar{\imath}, \bar{\beta}$ is a solution of \eqref{eq:full_dynamics}.
\end{Lemma}

\begin{proof}

Consider the linear change of coordinates
%
    $x = - (\imath + s)/\gamma$.
%
%
Adding up the two equations in \eqref{eq:full_dynamics}, yields $\dot{\imath} + \dot{s} = - \gamma \imath$. 
%
%
%
We can therefore establish the change of coordinates
\begin{equation}\small\label{eq:invert_change}
    \imath = \dot{x}, \quad s = - \gamma x - \dot{x}.
\end{equation}
Combining the latter, with the second equations in \eqref{eq:full_dynamics} allows writing the following equivalent formulation of the SIR dynamics
\begin{equation}\small\label{eq:total_x}
        \ddot{x} = - (\gamma x + \dot{x}) \dot{x} \beta - \gamma \dot{x}.
\end{equation}
%
%
We take the following control action
\begin{equation}\small\label{eq:control_x}
    \beta(x,\dot{x}) = 
    -\frac{\gamma \dot{\bar{x}} + \ddot{\bar{x}}}{(\gamma x + \dot{x})\dot{x}}
     + \alpha_{\mathrm{p}} ({\bar{x}} - x) + \alpha_{\mathrm{d}}  (\dot{\bar{x}} - \dot{x}),
\end{equation}
with $\alpha_{\mathrm{p}} > \gamma$, $\alpha_{\mathrm{d}} > 0$ being the gains of a PD-like action.
This produces the closed loop dynamics
%
   $ \ddot{e} = -(\gamma + \alpha_{\mathrm{d}} (- \gamma x - \dot{x})\dot{x}) \dot{e} - \alpha_{\mathrm{p}} (-\gamma x - \dot{x})\dot{x} e$, 
%
where $e = \bar{x} - x$. 
By hypothesis \cosimo{$\gamma + \alpha_{\mathrm{d}} (- \gamma x - \dot{x})\dot{x} > 0$ and $\alpha_{\mathrm{p}} (- \gamma x - \dot{x})\dot{x} > 0$}.
Therefore, both $e$ and $\dot{e}$ converge \cosimo{exponentially} to zero~\cite{calzolari2020exponential}, \cosimo{which 
in turn assures that $(s,\imath)$ converges exponentially to $(\bar{s},\bar{\imath})$}.
\cosimo{
We need to show now that \eqref{eq:controller_fb} and \eqref{eq:control_x} are equivalent.
First, 
we use \eqref{eq:total_x} to obtain $\ddot{\bar{x}} = (\gamma \bar{x} + \dot{\bar{x}})\dot{\bar{x}} - \gamma \dot{\bar{x}}$.
We then take $\psi_{\mathrm{s}} = \alpha_{\mathrm{p}}/\gamma$ and $\psi_{\mathrm{i}} = \alpha_{\mathrm{d}} - \alpha_{\mathrm{p}}/\gamma$.
Finally, we combine these three equations with \eqref{eq:invert_change} and \eqref{eq:control_x}. This leads to \eqref{eq:controller_fb}, therefore concluding the proof.
}
%
%
%

\end{proof}

We want our control action to remain limited when acting on a neighborhood of $s \imath = 0$. \cosimo{ Also, it is not  meaningful to act on the system by changing $\beta$ to values smaller than the one associated with total lockdown $\beta_{\mathrm{min}} > 0$, or greater than the one representing no social distancing $\beta_{\mathrm{max}} > \beta_{\mathrm{min}}$. We therefore introduce the following modification on the ideal controller}
\begin{equation}\small\label{eq:actual_control_x}
\color{black}
    \beta(s,\imath,t) = 
        \left[\psi_{\mathrm{i}} (\bar{\imath} - {\imath}) - \psi_{\mathrm{s}} (\bar{s} - {s}) +\frac{\bar{s}\bar{\imath}}{[s \imath]_{\epsilon}^{\infty}}\;\bar{\beta}\right]_{\beta_{\mathrm{min}}}^{\beta_{\mathrm{max}}}\!\!,
\end{equation}
where $\epsilon > 0$ is a small constant, and $[a]_{l}^{u}$ is is equal to $l$ or $u$ if $a < l$ or $a > u$ respectively, and equal to $a$ otherwise. 
%
%
\begin{figure}
	\centering
	\includegraphics[width=0.9\columnwidth]{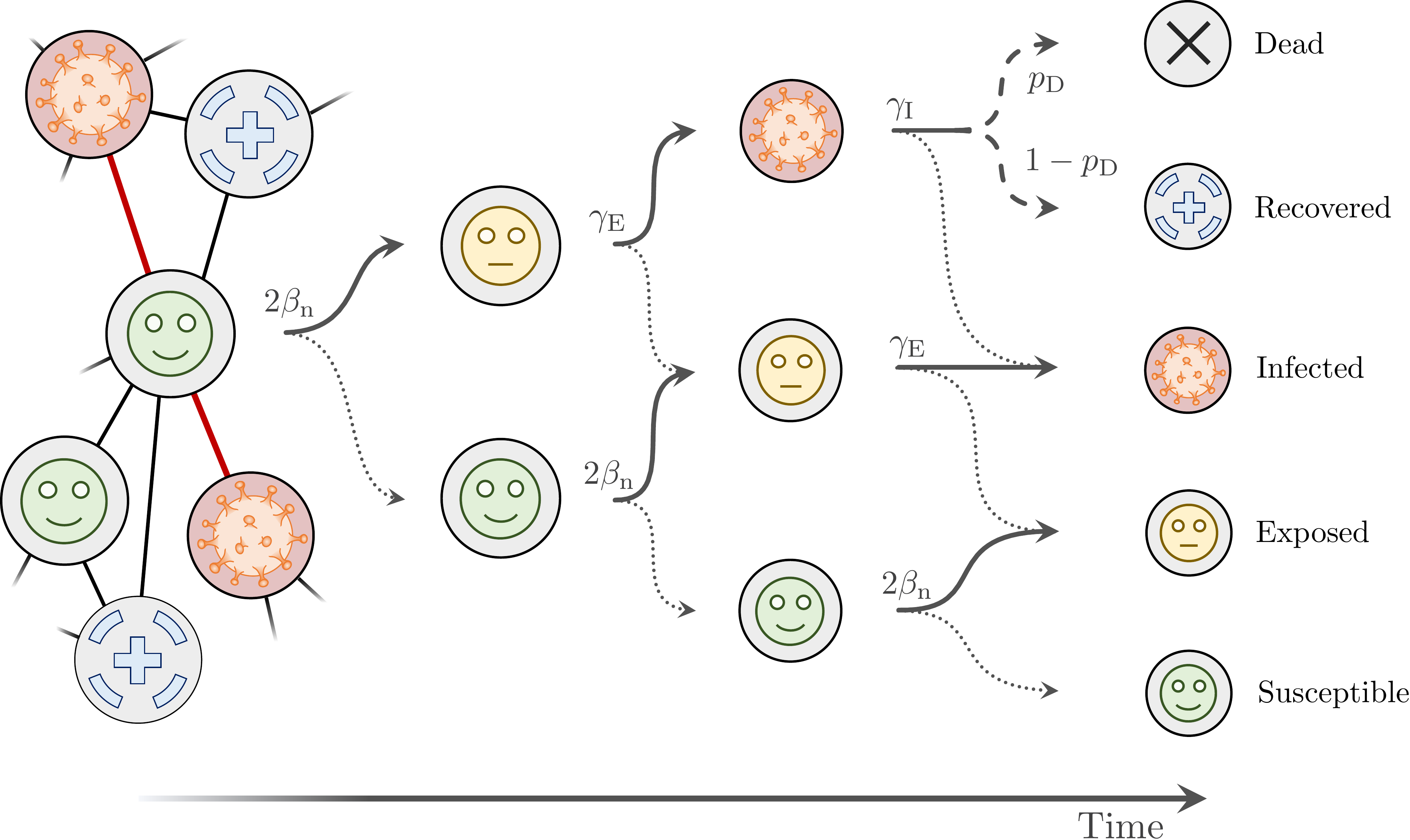}
	\caption{
	    \cosimo{Pictorial representation of SEIRD dynamics on a network. The  process is a continuous-time Markov chain. Each infected (and infectious) node spreads the disease to its susceptible neighbors at rate $\beta_{\mathrm{n}}$ until  no longer infectious. A node that has been successfully infected,  becomes first exposed, then infectious itself. Its ultimate destiny is either dying (with probability $p_{\mathrm{D}}$), or fully recovering (with probability $1-p_{\mathrm{D}}$). The rate of each event is given on the continuous arrows.}
		\label{fig:seird}
	}
\end{figure}
Fig. \ref{fig:track_traj} reports two examples of application of the algorithm to the SIR model \eqref{eq:full_dynamics}.
%


\section{Network Control}\label{sec:network}

\subsection{Network Model}

\fra{We implement two important features in a refined model:   (i) people interact through heterogeneous contact structures, i.e. the population is not well-mixed, and (ii) real epidemics have an intrinsic degree of stochasticity, so they cannot be exactly described by  \eqref{eq:full_dynamics}. We therefore consider stochastic epidemics on networks~\cite{pastor2015epidemic,Kiss2017}. A network is a pair $(V,E)$, where $V$ is a set of $N$ nodes (or vertices), and $E$ is a set of edges (or links) connecting nodes, i.e. tuples $\{u,v\}$, where $u,v \in V$. A population contact structure is modelled by a network in which nodes represent individuals, and  links are associated with routes of disease transmission between individuals. We consider undirected networks, such that $\{u,v\} \in E \iff \{v,u\} \in E$. Figs. \ref{fig:cover}, \ref{fig:seird}, \ref{fig:network_model}, show pictorial representations of networks.
Here, we focus on a particular well-known class of random networks, i.e.  Erd\H{o}s-R\'enyi~\cite{graph2001}, generated as follows: start with $N$ isolated nodes, consider each unique pair of two distinct nodes and connect them with probability $0\leq p \leq 1$. Hence, the probability of a node  having $k$ neighbors follows a binomial distribution $\mathcal{B}(N-1,p)$, $E(k)=p(N-1)$ being the average degree. Such networks may be considered a very first order approximation of realistic contact structures, as they display sufficient  heterogeneity and are easy to implement~\cite{Kiss2017}.
}

\subsection{Epidemic model on Network}

\fra{
We consider a SEIRD model for disease spreading, in which, at any time, each node has to be in one of five states  representing its status with respect to the disease: susceptible ($S$), exposed ($E$), infected/infectious  ($I$),  recovered ($R$) or deceased ($D$).
Fig.~\ref{fig:seird} illustrates  the possible transitions of a susceptible node that is in contact with two infectious neighbors.
Compared to a SIR model (see Sec. \ref{Sec:background}), we add an exposed class to account for individuals who have been infected but are  not yet infectious (biologically known as incubation phase).  We also allow for infected individuals to either survive or die.
Outbreaks are modeled as Markovian processes on the generated network, in which a node $I$ infects, via links, its $S$ neighbors  at a constant rate $\beta_{\mathrm{n}}$,  turning them in $E$. At a constant rate $\gamma_{E}$, an $E$ node becomes $I$. $I$ nodes stop being infectious independently at a constant rate $\gamma_{I}$, after which they have two possibilities: either they fully recover ($R$), with probability $1-p_{\mathrm{D}}$, or they die with probability $p_{\mathrm{D}}$ ($D$). Nodes in state $R$ and $D$ play no further role in the dynamics.  Further, $p_{\mathrm{D}}$ depends on the prevalence of the disease, to model  increased mortality in case of saturation of the health care system. Control interventions in this model are implemented as changes in the value of   $\beta_{\mathrm{n}}$. At time $t=0$, $I(0) = N \imath(0) \ll N$ randomly chosen nodes are infected. The remaining ones are initialized as susceptible. 
We use a Gillespie algorithm~\cite{Gillespie1977} adapted to networks~\cite{Kiss2017} to simulate this process. In Fig~\ref{fig:network_model} we show a realization of an outbreak on a network of modest size, to highlight how the topology impacts the dynamics.}
\subsection{Input and Output Maps}
\fra{ To connect the controller to the network model, we introduce two maps, as shown in Fig. \ref{fig:cover}. 
Such mappings are  general, and they could be used in conjunction with different control techniques relying on similar input\--output description of the pandemic.
The output map extracts $s$ and $\imath$ from the full state of the network by counting as $s$ the fraction of nodes either $S$ or $E$, and as $\imath$ the fraction of $I$. The input map provides expressions for the control input on the network level $\beta_{\mathrm{n}}$ given the output of the controller $\beta(s,\imath,t)$.
}
\fra{ 
With the aim of evaluating the input map, we turn to the adaptation of $\beta$ to networks}. From~\eqref{eq:full_dynamics} we get
\begin{equation}\small
    N\dot{\imath} = \beta N \imath s - \gamma N\imath \Rightarrow \dot{I} = \beta I \frac{S}{N} - \gamma I.
\end{equation}
The term $\beta I S / N$ represents the total infectious pressure in the ODE model. 
This quantity drives the whole infectious process, and it is crucial that the map preserves it. On the network, the infectious pressure is given by the infectious pressure $\beta_\mathrm{n}$  times the number of links between infected and susceptible nodes, which is a random variable that depends on which nodes are infected/recovered and on the topology of the network. \cosimo{Therefore, implementing an exact mapping would require to impose a different SD level on each individual, depending on the degree of its social interactions.
Although well defined in theory, this is clearly not  implementable in practice.
}
To overcome this issue, we introduce the so-called mean-field approximation~\fra{\cite{pastor2015epidemic,Kiss2017}}. On average, an infected node is connected to $E\left[k\right]$ neighbors, of which we assume that a proportion ${S}/{N}$ is susceptible. Hence, we set the number of $S-I$ links as $E\left[k\right] I {S}/{N}$. We derive $\beta_{\mathrm{n}}$ as a simple linear function of $\beta$
\begin{equation}\small
    \beta_{\mathrm{n}} I(t) E\left[k\right] \frac{S(t)}{N} \simeq \frac{\beta}{N} I(t) S(t) \Rightarrow \beta_{\mathrm{n}} \simeq \frac{\beta}{E\left[k\right]}.    \label{eq:controlmap}
\end{equation}
\fra{
This is a valid first-order approximation, that is known to give an upper estimate of the true $S-I$ link count~(see \cite{Kiss2017,pastor2015epidemic}), which in our case can only translate in a more conservative control strategy.}
\fra{ This expression connects a SIR model~\eqref{eq:full_dynamics} to a stochastic SIR on networks, rather than a stochastic SEIRD, as we want. Hence, we need to add an additional layer that conciliates $\gamma_{\mathrm{E}}$ and $\gamma_{\mathrm{I}}$ with $\gamma$ in the SIR model. To do so, we first consider  the time to full recovery (or death) of an individual who has been infected in a SEIRD model. This is a random variable exponentially distributed with rate $\frac{\gamma_{\mathrm{E}}\gamma_{\mathrm{I}}}{\gamma_{\mathrm{I}}+\gamma_{\mathrm{E}}}$. We set the controller $\gamma$ to this value. To find the infection rate, we use the definition of $R_0$~\cite{Kiss2017} for both models, i.e. $R_0 = \frac{\beta}{\gamma}$ for the SIR, and $R_0 = \frac{\tilde{\beta}}{\gamma_{\mathrm{I}}}$ for SEIRD (we momentarily use $\tilde{\beta}$ to distinguish it from the $\beta$ in the SIR), and we impose that they are equal. This yields
$\tilde{\beta} = \beta \frac{\gamma_{\mathrm{I}} + \gamma_{\mathrm{E}}}{\gamma_{\mathrm{E}}}$. Finally, combining this expression with \eqref{eq:controlmap}, gives $\beta_{\mathrm{n}}$ as}
\begin{equation}\small
\color{black}
    \beta_{\mathrm{n}} = \frac{\tilde{\beta}}{E\left[k\right]} =   \frac{\beta}{\gamma_{\mathrm{E}}} \frac{\gamma_{\mathrm{I}}+\gamma_{\mathrm{E}}}{ E\left[k\right]}.
    \label{eq:map}
\end{equation}

%

\begin{figure}[t]
    \centering
    \includegraphics[trim={2cm 3cm 3.5cm 1cm},clip,width=0.49\columnwidth]{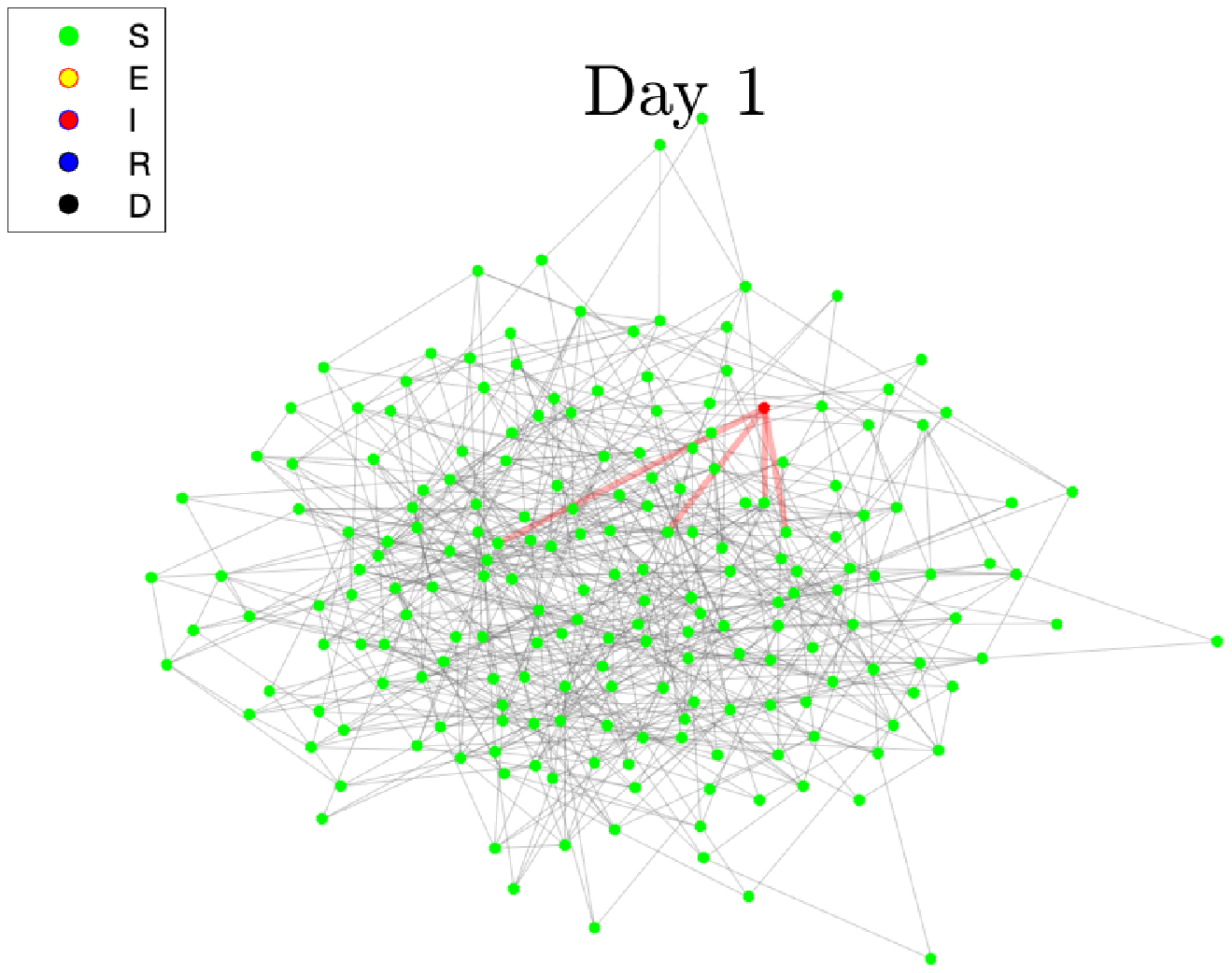}
    \includegraphics[trim={2cm 3cm 3.5cm 1cm},clip,width=0.49\columnwidth]{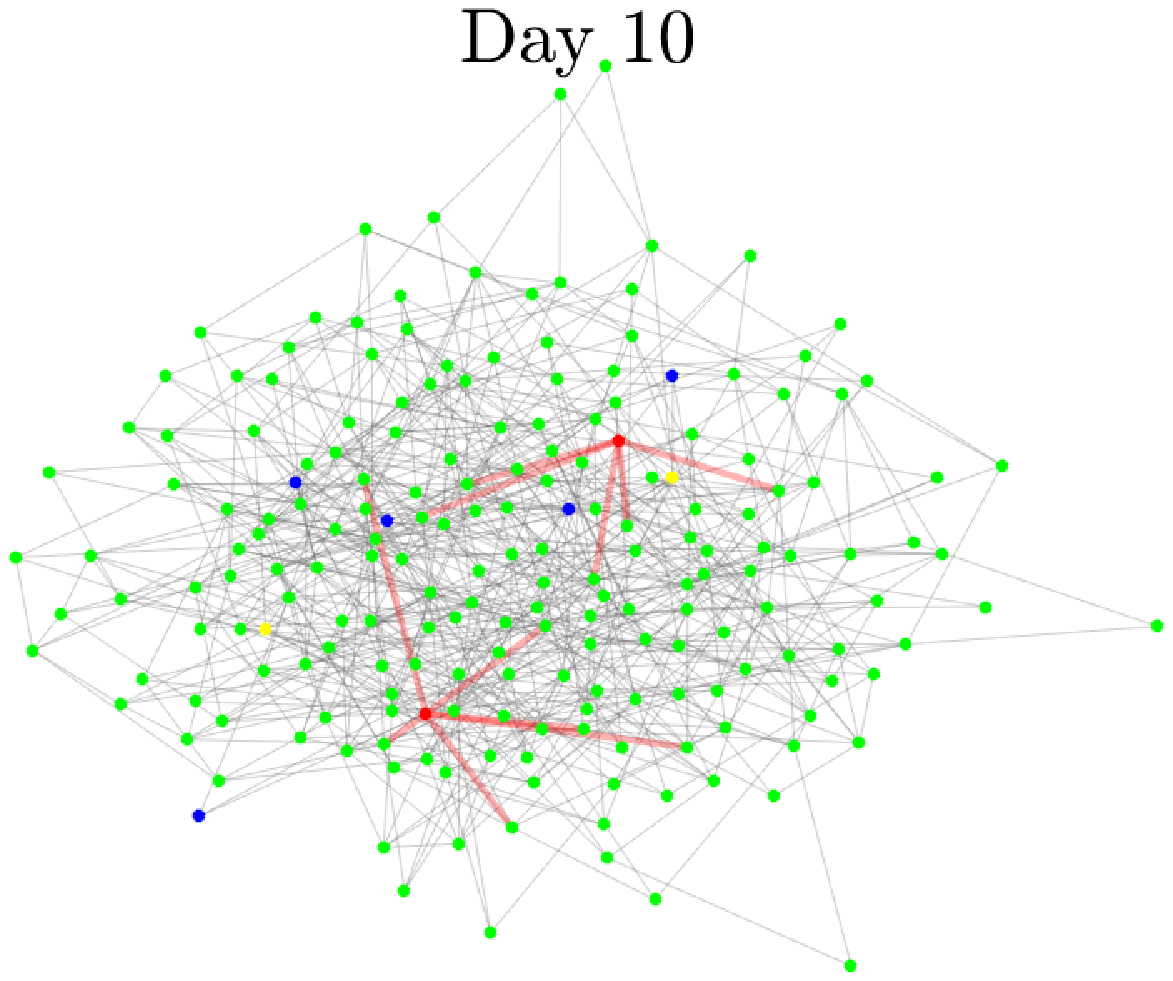}
    \includegraphics[trim={2cm 3cm 3.5cm 1cm},clip,width=0.49\columnwidth]{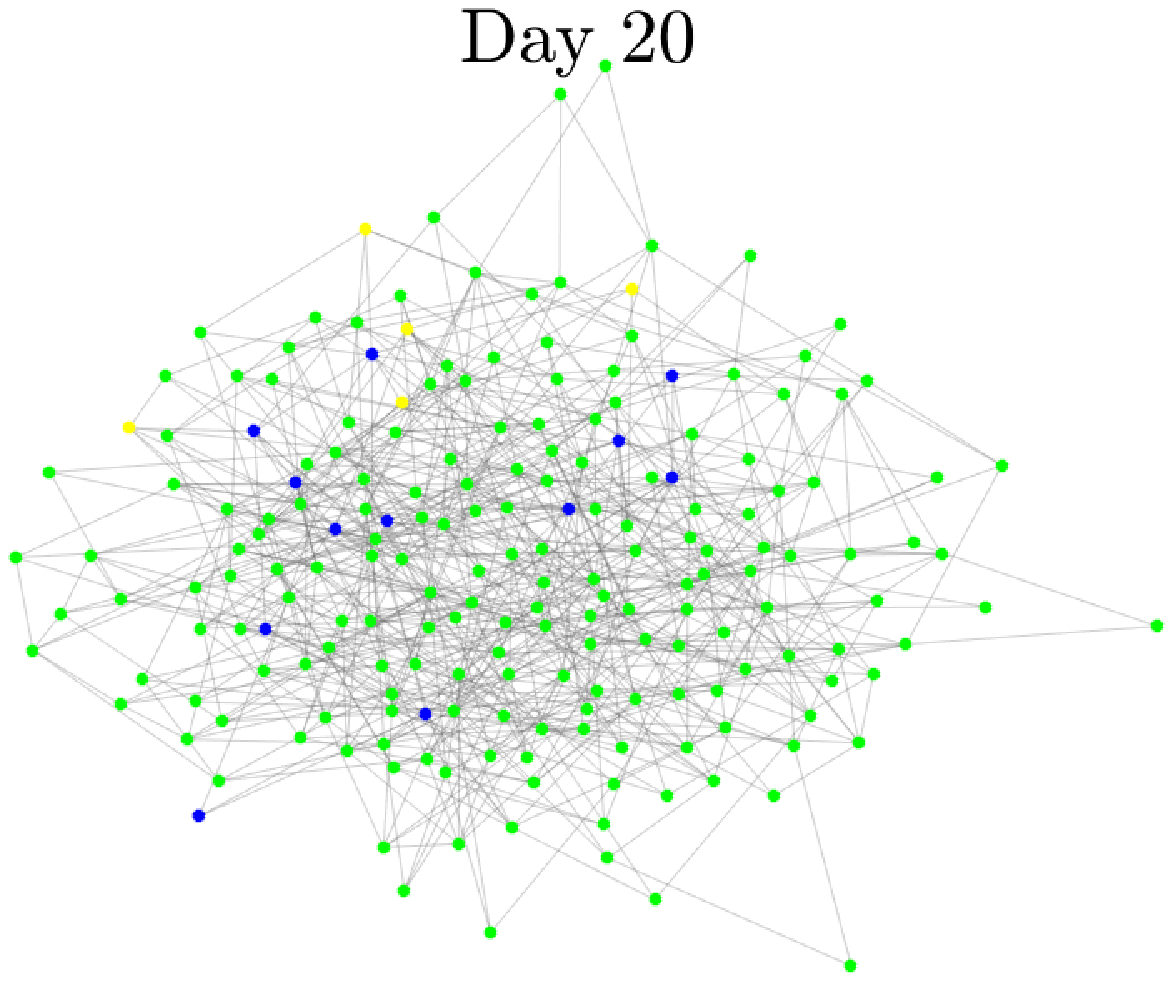}
    \includegraphics[trim={2cm 3cm 3.5cm 1cm},clip,width=0.49\columnwidth]{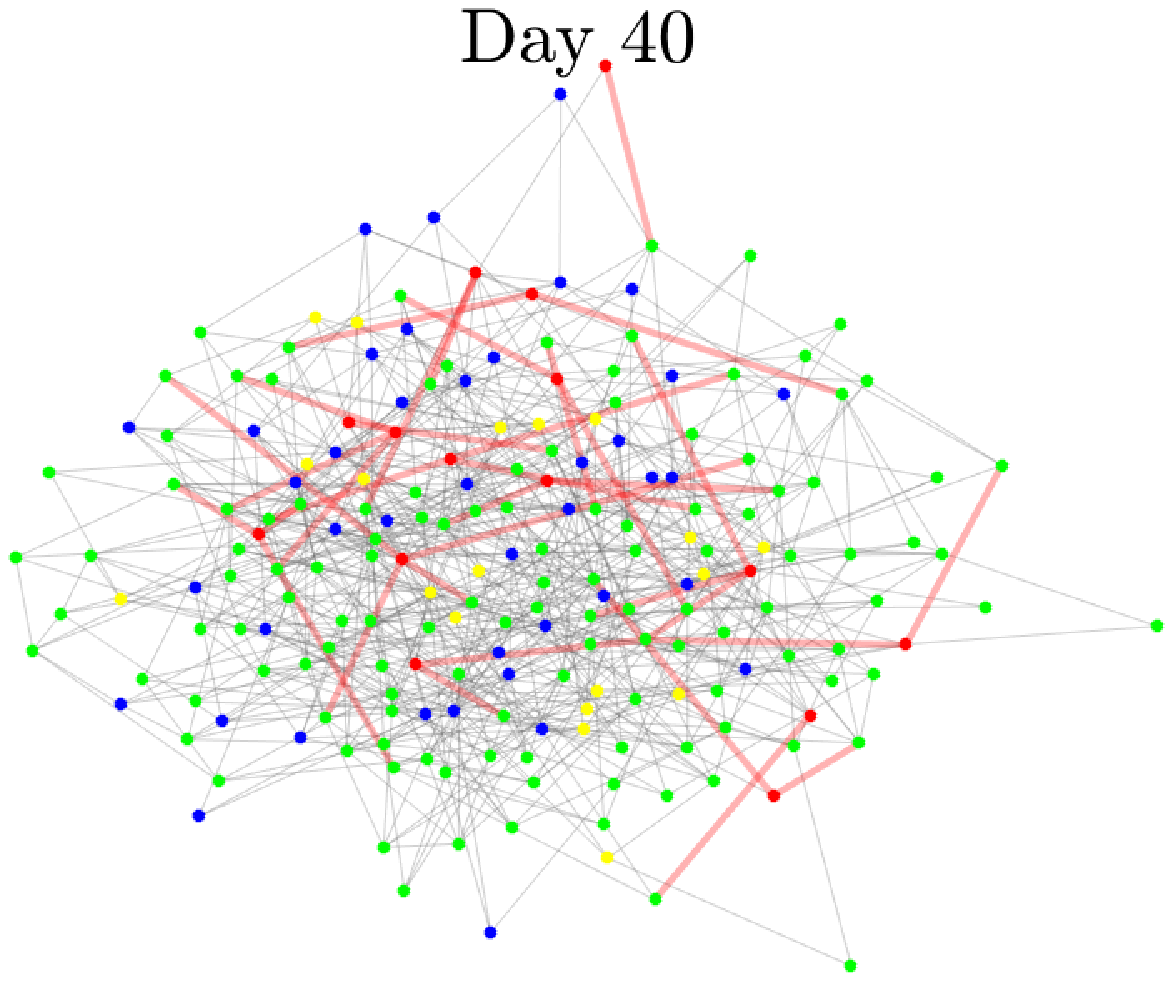}
    \caption{\fra{Simulation of a SEIRD  outbreak on a Erd\H{o}s-R\'enyi network of size $200$, with average degree $E\left[k\right] = 7$. A single node at day $1$ spreads infection to its neighbors (red edges), which in turn become first exposed, then infected, and eventually recover or die. The network is drawn in such a way that nodes with fewer links are on the periphery. The effect of the topology on the disease is particularly evident on such nodes, as only a few of them gets infected compared to  central ones.} 
    }
    \label{fig:network_model}
\end{figure}
%
%

\section{Simulations} \label{sec:validation}
%

%
\cosimo{
On top of the complexity introduced by the network dynamics, we consider several non\--ideal behaviors to better approximate a real\--world scenario. Note that none of these effects are considered in the controller design, and therefore are to be seen as uncertainties.
\begin{itemize}
    \item Unknown random delay affects measurements, which changes every time the controller is executed. This models the difficulties in getting on-line estimates of prevalence through daily swab tests.
    %
    \item  Policy update is allowed at a fixed rate, to mimic real life scenarios in which policy makers are reluctant to apply different degrees of restrictions too frequently.
    \item Quantization of the possible levels of $\beta$. Policy makers can realistically implement only limited control actions.
    We use $5$ distinct, equally spaced, levels, from $\beta_{\mathrm{min}}>0$ to $\beta_{\mathrm{n}}$. We set $\beta_{\mathrm{min}} = 0.25 \beta$. 
    This is based on 
    the analysis of Italian mobility data \cite{pepe2020covid}. 
    \item We introduce measurement noise of the signal, proportional to its value, to model uncertainty in the estimation of the prevalence when the epidemic is out of control.
\end{itemize}
}

\begin{table}
{
\centering
\caption{Parameters used for simulations in Sec.~\ref{sec:validation}.}
\label{table:params}
\begin{tabular}{ |l |c || l|  c| }
\hline
 $\beta_{\mathrm{n}}$ & $0.0227$ & $\beta_{\mathrm{min}}$  & $0.0057$\\
 \hline
 $\gamma_{\mathrm{E}}$ & $0.25$ & delay (days) & $\mathcal{N}(\{3,7,20\},1)$  \\ 
 \hline
 $\gamma_{\mathrm{I}}$ & $0.1428$ & noise (signal) &  $\mathcal{N}(0,0.1)$ \\ \hline
 $N$ & $16000$ & Hospitalization rate & $0.02$ \\
 \hline
 
 $T_f$ (days) & $240$ &  $\imath_{th}$ (\%) & $0.025$ \\
 \hline 
  $I_0$ & $800$ & $p_{\mathrm{D}}$ if  $ \imath \leq \imath_{th}$ & $0.005$\\
 \hline
   $S_0$ & $15200$ & $p_{\mathrm{D}}$ if $ \imath \leq \imath_{th}$  & $0.02$\\
\hline
  $\mathrm{E}[k]$ & $19$  & policy update  (days) & $\{1,7,15\}$ \\
 \hline
 
\end{tabular}
}
\end{table}
\cosimo{For the tuning of the model parameters we consider the case of Codogno, which}
has been the first city  in Lombardy with a diagnosed case of Covid-19. \fra{We have used Google data for the number of people in Codogno and the hospital capacity. We considered realistic parameters for incubation period~\cite{Ling2020}, infectious period~\cite{Ling2020}, hospitalization rate
~\cite{oms2020,jama2020}, infection fatality rate~\cite{Meye2020}, and social network connectivity~\cite{Melegaro2011}. All the parameters are reported in Table~\ref{table:params}. The initial condition is set to $I_0 = 800$, to model a delayed recognition of the presence of the disease, and simulations are run for $T_f = 240$ days.}

\fra{
Figs. \ref{fig:codogno_infected} and \ref{fig:codogno_beta} show the evolution of infected $\imath$, deaths $D/N$, and prescribed SD $\beta$, for the case where policy can change once a week and delay between testing and results is on average $4$ days. 
We report the results when using the proposed feedback action $\beta(t, s,\imath)$ and, as comparison, the evolution of the uncontrolled epidemics ($\beta = \beta_{\mathrm{max}}$) and of a one on-off intervention lasting for $60$ days, during which $\beta = \beta_{\mathrm{min}}$.} 
Susceptible percentages $s$ are not shown for the sake of space. 
We aggregate results from $100$ simulations, each one run on a different network realization.
%
%

\begin{figure}[t]
    \centering
    \subfigure[Infected $\imath$]{\includegraphics[width=0.95\columnwidth,trim= {0.8 0cm 1cm 0}, clip]{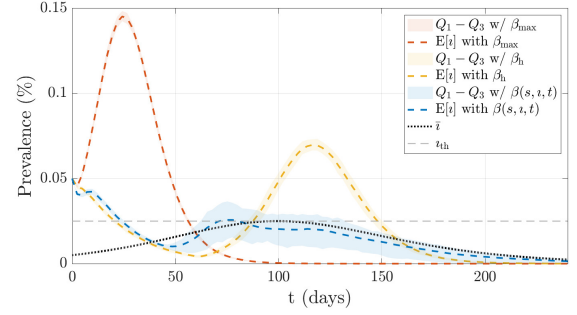}}
    \subfigure[Deaths $D/N$]{\includegraphics[width=0.95\columnwidth,trim= {0.8 0cm 1cm 0}, clip]{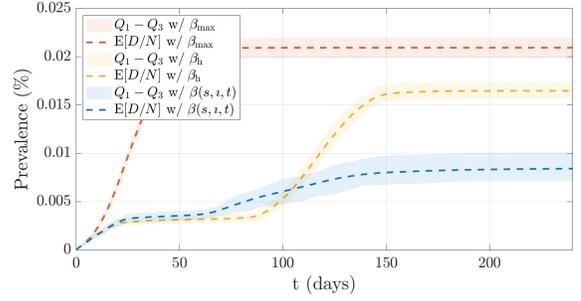}}
    %
    %
    %
    \caption{Prevalence of  infected and dead nodes for the considered simulation scenario. It is shown here the case in which the policy changes only once every week, and the average delay in measurements is set to $3$ days. All the other values are as in Tab. \ref{table:params}. } 
    \label{fig:codogno_infected}
    \hspace{-2cm}
\end{figure}
\begin{figure}[t]
\centering
    \includegraphics[width=0.95\columnwidth,trim= {0.8cm 0cm 1cm 0}, clip]{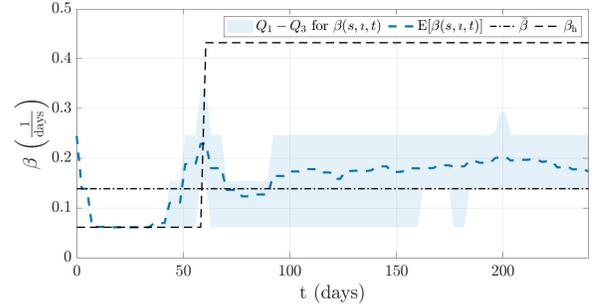}
    %
    %
    %
    \caption{
    \cosimo{
    Level of SD $\beta$ as a function of time. The average output of the controller across 100 simulations, when policy changes every $7$ days and delays in data are $3$ days, is shown together with its lower and upper quartiles (Q1-Q3).
    We also report for comparison a $60$ days full\--lockdown strategy, and the feedforward action $\bar\beta$.
    }
    }
    \label{fig:codogno_beta}
\end{figure}

\fra{We evaluate the performances of the controller in various settings, in which we act on two main parameters, namely, the delay in knowledge of the status and the frequency at which the control policy can be changed. The former one can take values of $\{3,7,20\}$ days, while the latter moves between $\{1,7, 15\}$ days. We consider all the possible combinations of these parameters. }
%
%
We cannot report here the complete results of our simulations, for the sake of space.
We report instead some relevant performance indexes \fra{in Fig.~\ref{fig:heat_maps} - namely the reduction in social distancing compared to $60$ days full lockdown, and reduction in deaths with respect to not applying any strategy. We observe 
that the controller performs well on average even in the most extreme cases. Yet, we observe increased dispersion as we increase delays and reaction times. The use of the controller consistently induces a reduction of over $32\%$ of deaths in the worst case, and, in the best tested case, of $63\%$.}

\section{Discussion}

Our approach resulted in a strategy able to keep the curve below the health care capacity when uncertainty is low, with increased variability when delays and other inaccuracies in measuring become important.
\fra{
%
%
From this analysis, it appears clear that is crucial to have a reliable estimate of the current prevalence of the disease. This is of course the downside of closed loop strategies, i.e. that the controller becomes less reliable as the quality of measurements deteriorates. Instead, it is worth noting that, given low delays in data, updating policies every $15$ days has a limited impact on the performance of the controller. Interestingly, increasing delays (or control updates frequency), does not have a major impact on the average performance of simulations, in terms of reduction of mortality. However, this result might be misleading, because the variance between different realisations gets higher as the delay increases, meaning that the controller becomes unreliable if applied to an individual realisation. This suggests that the crucial quantity  for control is on-line prevalence estimation.
Therefore, this analysis confirms that, when implementing control policies based on daily testing data, policy makers should ideally have access to the exact state of the system. Clearly, this is far from being a realistic assumption. Still, our results prove that periodic loop closure can still be a viable solution also in a more realistic setting - although we are in no position to claim any definitive result in this direction. 
}

At the same time, we observe a relevant outcome in all our simulations, namely that when control acts on an outbreak that has already reached a significant proportion of the population, the advisable strategy is to go into full lockdown until the epidemic curve is brought down to acceptable levels, and then to gradually relax and adjust  control measures, according to the estimated prevalence. 

	%

\begin{figure}[t]
    \centering
    \subfigure[Reduction in $\beta$, for $\imath(0) \neq \bar{\imath}(0)$]{\includegraphics[width = 0.485\columnwidth, trim = {0 1.25cm 0 0}, clip]{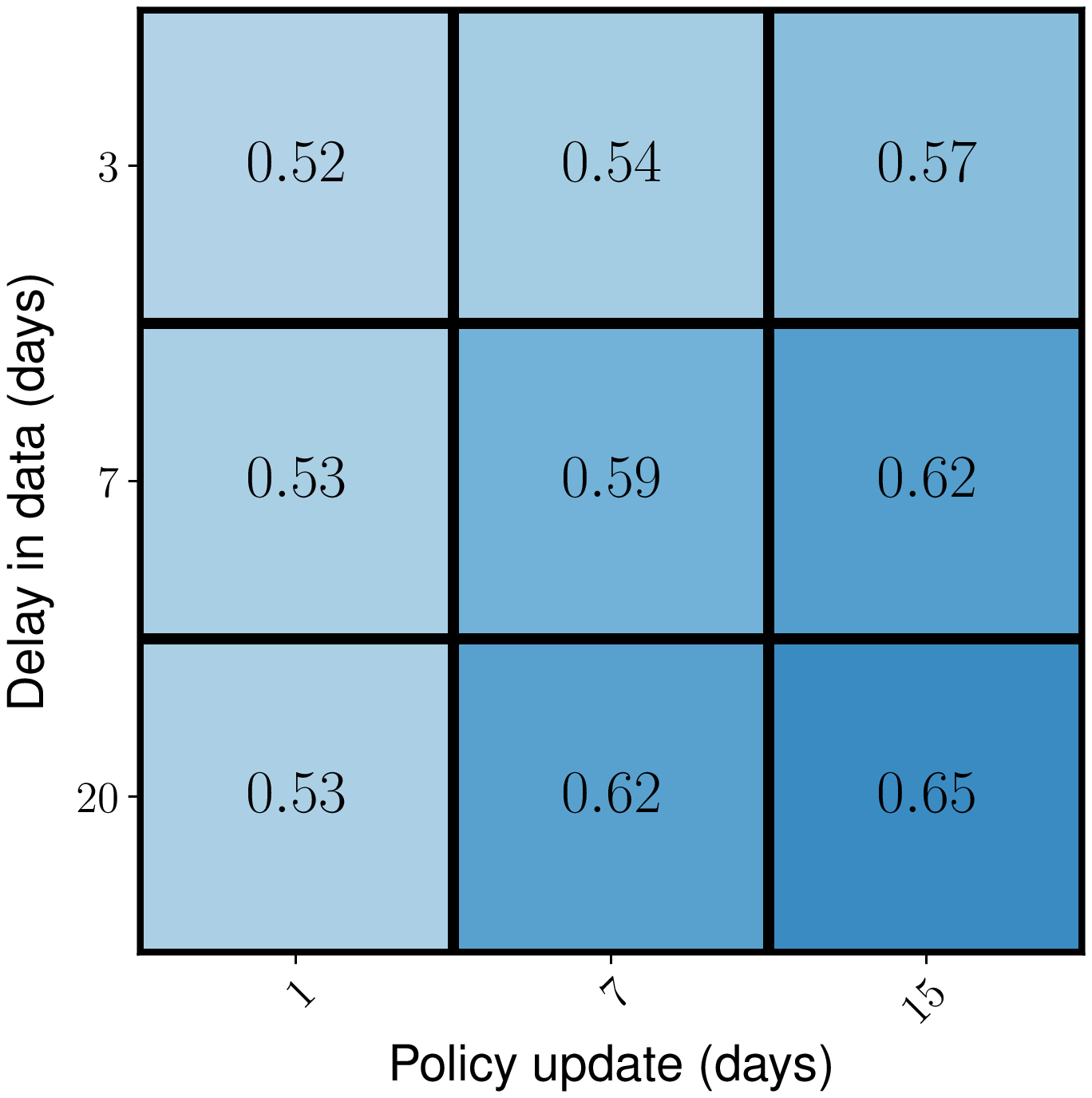}}
    \hspace{0.0025\columnwidth}
    \subfigure[Reduction in $D$, for $\imath(0) \neq \bar{\imath}(0)$]{\includegraphics[width = 0.485\columnwidth, trim = {0 1.25cm 0 0}, clip]{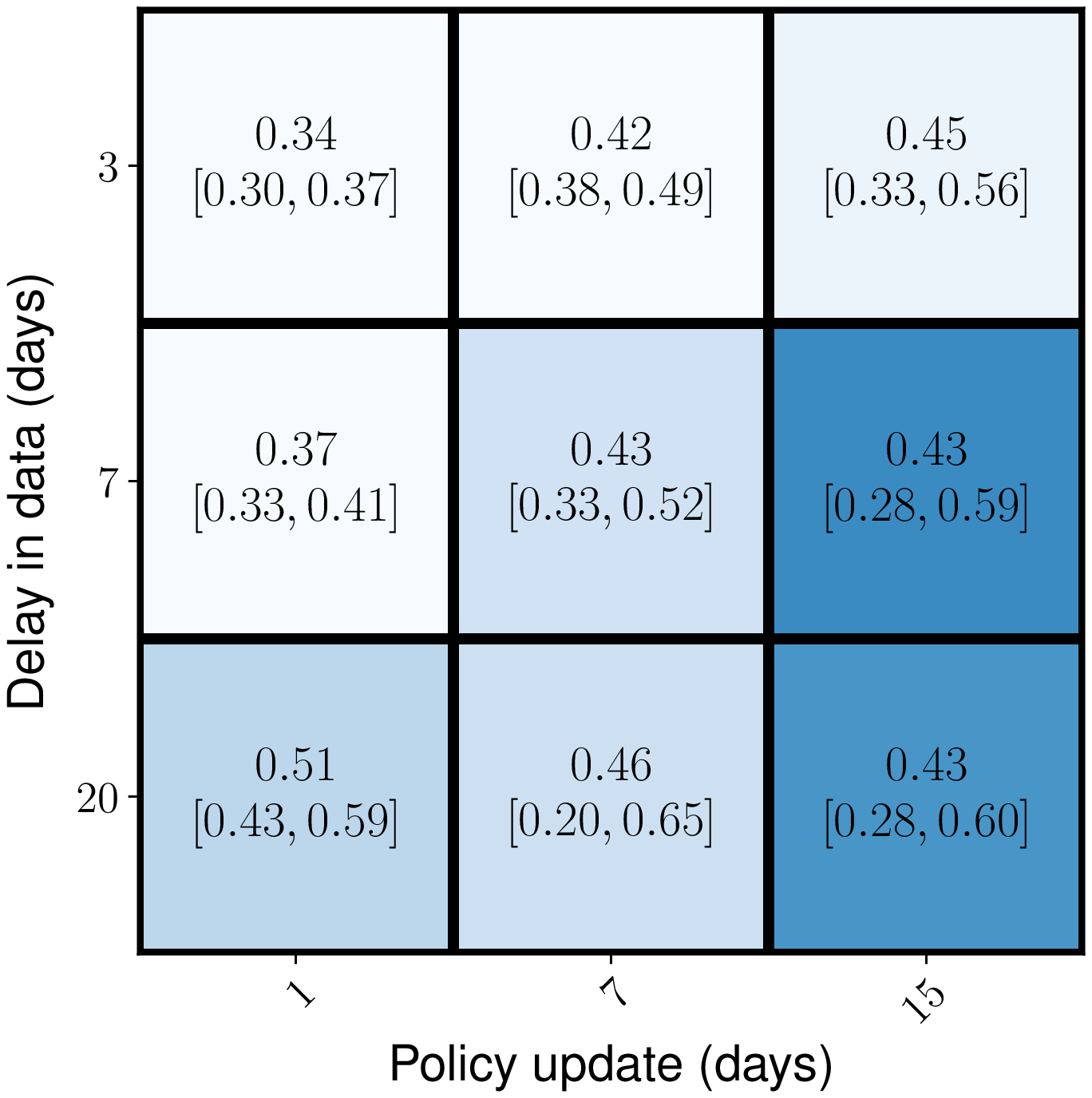}}
    \subfigure[Reduction in $\beta$, for $\imath(0) = \bar{\imath}(0)$]{\includegraphics[width = 0.485\columnwidth, trim = {0 1.25cm 0 0}, clip]{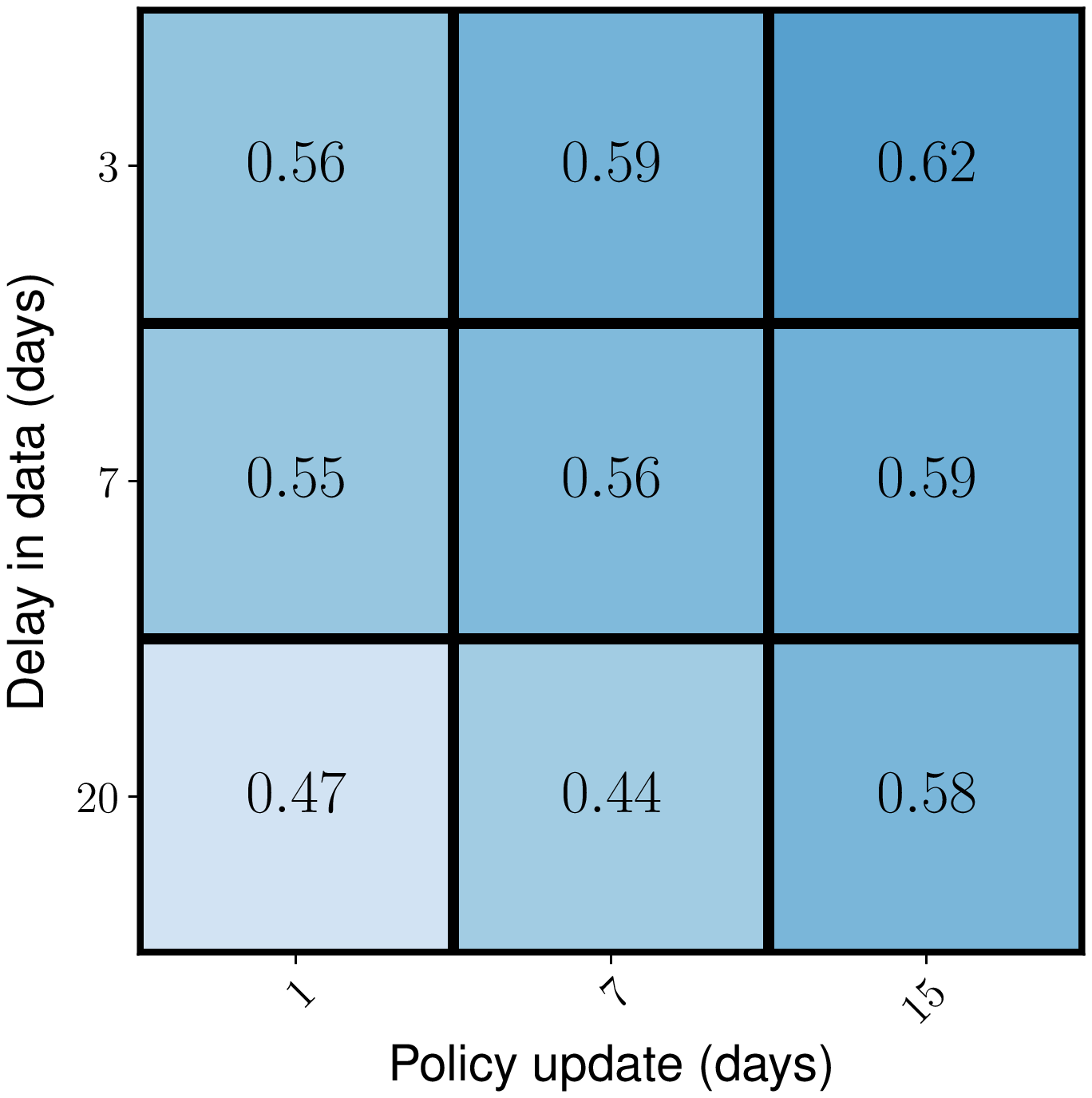}}
    \hspace{0.0025\columnwidth}
    \subfigure[Reduction in $D$, for $\imath(0) = \bar{\imath}(0)$]{\includegraphics[width = 0.485\columnwidth, trim = {0 1.25cm 0 0}, clip]{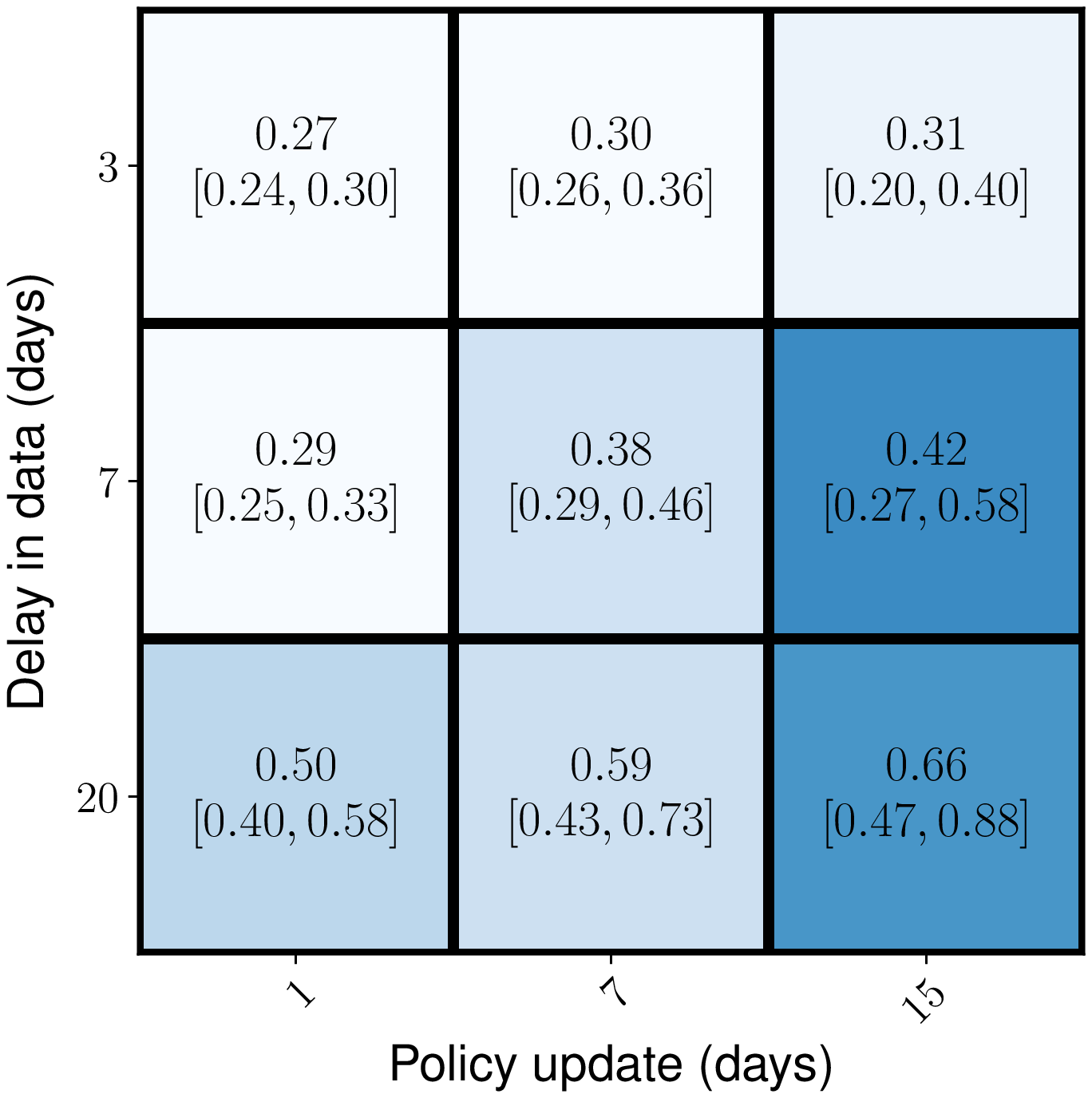}}
    \caption{\fra{Heat maps reporting (a) the average reduction in $\beta$, normalised by $\int_0^{T_f} \beta_{\mathrm{h}}(t) \mathrm{d}t$ in the reference scenario, and (b) the average reduction in deaths, normalised by the average number of deaths in the reference scenario, across different conditions.  Colors in (b) follow the width of the fist and third percentile (reported in the cells under the average).}
    \cosimo{Both the indices are defined so that the smaller the better.}
    }
    \label{fig:heat_maps}
\end{figure}

\section{Conclusions and Future Work}

This preliminary work showed that a simple feedback action can improve the robustness and the effectiveness of an optimal policy for epidemic control, even in presence of quite non ideal behaviors in the system and in  measuraments.
\cosimo{The effectiveness of strategies based on control for dealing with epidemics is still an open topic, with respected academics having opposite positions~\cite{casella2020can,nowzari2016analysis}. We do not aim here to give a final solution to the problem. On the contrary, we want to give our perspective to this important discussion by providing a new piece to this intricate puzzle.}
%
Future work will be devoted to use more reliable input maps (and possibly theoretical models for the controller), improve control design with robust and adaptive techniques, include other sources of lags and uncertainties, use more realistic network models - possibly dynamic networks, \fra{the ultimate goal being engineering a sound model that could be useful when it comes to decision making for  governments.} 

\bibliographystyle{ieeetr}%
\bibliography{biblio}

\end{document}